\documentclass[showpacs,showkeys,preprintnumbers,amsmath,amssymb]{revtex4}

\usepackage{graphicx}
\usepackage{dcolumn}
\usepackage{bm}
\usepackage{mathrsfs}
\usepackage{amsmath}
\newtheorem{theorem}{Theorem}[section]

\newtheorem{corollary}[theorem]{Corollary}
\newtheorem{remark}[theorem]{Remark}
\newenvironment{proof}[1][Proof.]{\begin{trivlist}
     \item[\hskip \labelsep {\bfseries #1}]}{\end{trivlist}}

\begin{document}


\title{Potentials of the Heun class}

\author{D. Batic}
\email{davide.batic@uwimona.edu.jm}
\author{R. Williams}%
\email{runako.williams@mymona.uwi.edu}
\affiliation{%
Department of Mathematics,\\  University of the West Indies, Kingston 6, Jamaica 
}%
\author{M. Nowakowski}
\email{mnowakos@uniandes.edu.co}
\affiliation{
Departamento de Fisica,\\ Universidad de los Andes, Cra.1E
No.18A-10, Bogota, Colombia
}%

\date{\today}

\begin{abstract}
We review different methods of generating potentials such that the one-dimensional Schr\"{o}dinger equation (ODSE) can be transformed 
into the hypergeometric equation. We compare our results with previous studies, and complement the subject with new findings.
Our main result is to derive new classes of 
potentials such that the ODSE can be transformed into 
the Heun equation and its confluent cases. 
The generalized Heun equation is also considered. 
\end{abstract}

\maketitle

\section{\label{Intro} Introduction}
Since the early days of Quantum Mechanics the problem of finding potentials such that the one-dimensional Schr\"{o}dinger 
equation (ODSE) 
\begin{equation}\label{schI}
\frac{d^2 u}{dx^2}+\left[k^2-V(x)\right]u(x)=0,\quad k^2=\frac{2mE}{\hbar^2},\quad x\in(-\infty,\infty)
\end{equation}
admits exact solutions attracted the attention of several physicists and mathematicians \cite{Levai}. The 
importance of the problem is directly related to the fact that an exact knowledge of the energy spectrum and of the 
wave functions of the ODSE turns out to be very useful in many applied problems \cite{Fl,Landau}. It is well-known that the ODSE can 
be solved analytically for a large number of potentials provided that we are able to find a coordinate transformation $y=y(x)$ and 
eventually a change of the dependent variable $u$ such that (\ref{schI}) can be transformed into a linear ordinary differential equation 
of mathematical physics. Here, we begin by reviewing the problem of reducing the ODSE to the hypergeometric equation
\begin{equation}\label{dueI}
y(1-y)\frac{d^2 v}{dy^2}+\left[c-(a+b+1)y\right]\frac{dv}{dy}-abv(y)=0,\quad y\in I\subset\mathbb{R}
\end{equation}
with $a,b,c\in\mathbb{C}$. \cite{Manning,Bhatta,Bose,Lam} considered only the case when the potentials depend 
explicitly on the variable $x$. This scenario has been further extended by \cite{Iwata,Natanzon}, where more general classes 
of  potentials exhibiting an implicit dependence on the coordinate $x$ have been derived. Note that the so-called Ginocchio potentials 
\cite{Gin} are a special case of the more general class of solvable potentials obtained by \cite{Natanzon}.  Moreover, \cite{Milson} derived a large 
class of potentials which include the Natanzon class by considering the case when a second order linear differential equation can be 
connected to a ODSE by means of Liouville transformations. We emphasize that the work of \cite{Milson} does not cover the case of the Heun 
equation, since the author assumed that the Bose invariant is a rational function with numerator represented by a 
polynomial of degree two or less. As we will see in Section~\ref{heunclass} such a numerator becomes a polynomial of degree four in the 
case of the Heun equation. Our main contribution is the derivation of the most general potential such that the ODSE can be transformed into the Heun equation \cite{Ronveaux}
\begin{equation}\label{Heun}
\frac{d^2 v}{dy^2}+\left(\frac{\gamma}{y}+\frac{\delta}{y-1}+\frac{\epsilon}{y-a}\right)\frac{dv}{dy}
+\frac{\alpha\beta y-q}{y(y-1)(y-a)}v(y)=0,\quad y\in I\subset\mathbb{R},
\end{equation}
where $a\in\mathbb{C}\backslash\{0,1\}$, $\epsilon=\alpha+\beta+1-\gamma-\delta$, and $\alpha$, $\beta$, $\gamma$, $\delta$, and $q$ 
are arbitrary complex parameters. For the role played by (\ref{Heun}) in General Relativity we refer to \cite{dav0}. Since the hypergeometric equation is a 
special case of the Heun equation with $q=\alpha\beta a$ and $\epsilon=0$, our potential generalizes the Natanzon class. We also 
treat the same problem for the confluent, biconfluent and triconfluent Heun equation. Moreover, we also find the most general class of 
potentials such that the ODSE can be transformed into the generalized Heun equation (GHE) \cite{ScSc}
\begin{equation}\label{GHE}
\frac{d^2 v}{dy^2}+\left(\frac{1-\mu_{0}}{y}+\frac{1-\mu_{1}}{y-1}+\frac{1-\mu_{2}}{y-\widehat{a}}+\widehat{\alpha}\right)
\frac{dv}{dy}+\frac{\beta_{0}+\beta_{1}y+\beta_{2}y^{2}}{y(y-1)(y-\widehat{a})}v(y)=0,
\end{equation}
where $\widehat{a}\in\mathbb{C}\backslash\{0,1\}$ and $\mu_{0}$, $\mu_{1}$, $\mu_{2}$, $\alpha$, $\beta_{0}$, $\beta_{1}$, $\beta_{2}$ 
are arbitrary complex numbers. To underline the importance of equation (\ref{GHE}) we recall that it contains the ellipsoidal wave 
equation as well as the Heun equation ($\alpha=\beta_{2}=0$) and thus the Mathieu, spheroidal, Lam$\acute{\mbox{e}}$, Whittaker-Hill 
and Ince equations as special cases. For physical applications of (\ref{GHE}) we refer to \cite{dav1}. Our work is organized as 
follows. In Section~\ref{iwa} we outline the derivation of the potential found by \cite{Iwata}, extend the findings thereof by computing two 
new classes of potentials and derive a non trivial condition under which the potential is symmetric. In Section~\ref{nata} we 
shortly review the derivation of Natanzon's potential and show that all Iwata's potentials are contained in the Natanazon class. 
In Section~\ref{heunclass} we compute the most general potential such that ODSE can be reduced to the Heun equation. 
This potential will be called potential of the Heun class. Starting from this result we also derive expressions for the potentials 
when the ODSE can be transformed into a confluent, doubly confluent, biconfluent and triconfluent 
Heun equation. We conclude our work by deriving the most general potential such that the ODSE can be transformed into the generalized 
Heun equation. Since the Heun equation is a special case of the GHE, we obtain a further generalization of the potentials belonging 
to the Heun class. A comment on the choice of parameters entering in the potentials derived in the next sections is in order. From here on, 
we will be concerned with real parameters in agreement with the choice of a real potential in the ODSE. However, the choice of complex 
parameters would also be possible in the ODSE when the potential is complex, a case which has some relevance in the physical literature and goes 
under the name of optical potential.

\section{Iwata's potential}\label{iwa}
We start as in \cite{Iwata} by restricting our attention to those potentials $V$ such that (\ref{schI}) can be reduced to (\ref{dueI})
after a suitable transformation of the coordinate $x$ and of the dependent variable $u$. If we transform (\ref{dueI}) according to
\begin{equation}\label{treI}
 y=y(x),\quad v(y)=s(x)u(x),
\end{equation}
we obtain
\begin{equation}\label{quattroI}
\frac{d^2 u}{dx^2}+\left[2\frac{s^{'}}{s}-\frac{y^{''}}{y^{'}}+\frac{c-(a+b+1)y}{y(1-y)}y^{'}\right]\frac{du}{dx}+
\left[\frac{s^{''}}{s}-\frac{y^{''}}{y^{'}}\frac{s^{'}}{s}+\frac{c-(a+b+1)y}{y(1-y)}\frac{s^{'}}{s}y^{'}-
\frac{ab(y^{'})^2}{y(1-y)}\right]u(x)=0. 
\end{equation} 
In the process of reducing (\ref{quattroI}) to (\ref{schI}) we must require that
\begin{equation}\label{cinqueI}
2\frac{s^{'}}{s}-\frac{y^{''}}{y^{'}}+\frac{c-(a+b+1)y}{y(1-y)}y^{'}=0.
\end{equation}
At this point (\ref{cinqueI}) can be used to eliminate the function $s$ in the last bracket in (\ref{quattroI}) and we 
end up with the equation
\[
\frac{d^2 u}{dx^2}+U(x)u(x)=0,
\]
where
\begin{equation}\label{doppiastella}
U(x)=T(x)+A(a,b)X(x)+B(a,b,c)Y(x)+C(c)Z(x) 
\end{equation}
with
\begin{equation}
\label{XYZ}
T(x)=\left(\frac{y^{''}}{2y^{'}}\right)^{'}-\left(\frac{y^{''}}{2y^{'}}\right)^2,\quad X(x)=\frac{(y^{'})^2}{4(1-y)^2},\quad Y(x)=\frac{(y^{'})^2}{4y(1-y)^2},\quad Z(x)=\frac{(y^{'})^2}{4y^2(1-y)^2},
\end{equation}
and
\begin{equation}\label{ABC}
A(a,b)=1-(a-b)^2,\quad B(a,b,c)=2[c(a+b-1)-2ab],\quad C(c)=2c-c^2.
\end{equation}
Since the parameters $a$, $b$ and $c$ may depend on the quantity $k$ defined in (\ref{schI}), we can assume without loss of generality that 
\begin{equation}\label{U}
U(x)=T(x)+A(k)X(x)+B(k)Y(x)+C(k)Z(x).
\end{equation}
Differentiating the above expression first with respect to $k$ and then with respect to $x$, 
we end up with the auxiliary equation
\begin{equation}\label{tilde}
\dot{A}X^{'}+\dot{B}Y^{'}+\dot{C}Z^{'}=0,
\end{equation}
where $\dot{}=d/dk$ and ${}^{'}=d/dx$. We must distinguish among the following three cases
\begin{enumerate}
\item 
The functions $X^{'}$, $Y^{'}$ and $Z^{'}$ are linearly independent. Then, $\dot{A}=\dot{B}=\dot{C}=0$ meaning that $A$, $B$ and $C$ are 
numerical constants. This case is not of practical relevance.
\item
Only one function is linearly independent. For instance, there exist constants, say $\rho$ and $\sigma$ such that 
$Y^{'}=\rho X^{'}$ and $Z^{'}=\sigma X^{'}$. A simple integration gives $Y=\rho X+c_1$ and $Z=\sigma X+c_2$. On the 
other side by means of (\ref{XYZ}) we find that the coordinate transformation $y$ should satisfy at the same time the 
differential equations
\[
\frac{1}{4(1-y^2)}\left(\frac{1}{y}-\rho\right)(y^{'})^2=c_1,\quad 
\frac{1}{4(1-y^2)}\left(\frac{1}{y^2}-\sigma\right)(y^{'})^2=c_2.
\]
This is the case if and only if $y^{'}=0$, but then $y$ is a constant function. Hence, this case must be disregarded.
\item
Two of the functions $X^{'}$, $Y^{'}$ and $Z^{'}$ are linearly independent. We have three distinct cases. 
\end{enumerate}
Let $Y^{'}$ and $Z^{'}$ be linearly independent. Then, there exist constants $\rho$ and $\sigma$ such that
\begin{equation}\label{duepunti}
X^{'}+\rho Y^{'}+\sigma Z^{'}=0. 
\end{equation}
By means of (\ref{duepunti}) we can rewrite (\ref{tilde}) as follows
\[
(\dot{B}-\rho\dot{A})Y^{'}+(\dot{C}-\sigma\dot{A})Z^{'}=0. 
\]
Since $Y^{'}$ and $Z^{'}$ are linearly independent, it must be $\dot{B}-\rho\dot{A}=0$ and $\dot{C}-\sigma\dot{A}=0$. Integrating 
these equations with respect to $k$ we obtain
\begin{equation}\label{stellad}
B=\rho A+\rho_1,\quad C=\sigma A+\sigma_1,\quad\rho_1,\sigma_1\in\mathbb{R}
\end{equation}
and $U$ can be expressed as follows
\begin{equation}\label{puntoesclamativo}
U(x)=T(x)+A\left[X(x)+\rho Y(x)+\sigma Z(x)\right]+\rho_1 Y(x)+\sigma_1 Z(x).
\end{equation}
If we integrate (\ref{duepunti}) with respect to $x$ and make use of (\ref{XYZ}), we discover that the coordinate 
transformation $y$ must satisfy the following first order nonlinear autonomous differential equation, namely
\begin{equation}\label{eq_y}
\frac{(y^{'})^2 R(y)}{4y^2 (1-y)^2}=\kappa, \quad\kappa\in\mathbb{R}\backslash\{0\},\quad R(y)=y^2+\rho y+\sigma.
\end{equation}
Taking into account that $U=k^2-V_I$ and applying (\ref{puntoesclamativo}), we find that
\[
V_I(x)=k^2+\left(\frac{y^{''}}{2y^{'}}\right)^2-\left(\frac{y^{''}}{2y^{'}}\right)^{'}
-A\left[X(x)+\rho Y(x)+\sigma Z(x)\right]-\rho_1 Y(x)-\sigma_1 Z(x).
\]
With the help of (\ref{XYZ}) and (\ref{eq_y}) it is straightforward to check that
\[
A\left[X(x)+\rho Y(x)+\sigma Z(x)\right]+\rho_1 Y(x)+\sigma_1 Z(x)=
\kappa\left[1-(a-b)^2\right]+\frac{\kappa(\rho_1 y+\sigma_1)}{R(y)}.
\]
Since the potential cannot depend on $k$, the parameters $a$ and $b$ must satisfy the following condition
\begin{equation}\label{zero}
(a-b)^2=-\frac{k^2}{\kappa}
\end{equation}
and the expression for the potential reduces to
\[
V_I(x)= \left(\frac{y^{''}}{2y^{'}}\right)^2-\left(\frac{y^{''}}{2y^{'}}\right)^{'}
-\kappa\frac{y^2+(\rho+\rho_1)y+\sigma+\sigma_1}{R(y)}.
\]
The first two terms on the r.h.s. of the previous expression can be computed as follows. First of all, let
\begin{equation}\label{f}
g(y)=\frac{4\kappa y^2(1-y)^2}{R(y)}. 
\end{equation}
Then, (\ref{eq_y}) can be written in the compact form $(y^{'})^2=g(y)$ from which we can easily derive the useful identities 
\[
y^{''}=\frac{1}{2}\frac{dg}{dy},\quad \frac{y^{'''}}{y^{'}}=\frac{1}{2}\frac{d^2 g}{dy^2}. 
\]
Taking into account that
\[
\left(\frac{y^{''}}{2y^{'}}\right)^2-\left(\frac{y^{''}}{2y^{'}}\right)^{'}=\frac{3}{4}\left(\frac{y^{''}}{y^{'}}\right)^2-
\frac{1}{2}\frac{y^{'''}}{y^{'}}=\frac{3(y^{''})^2}{4g}-\frac{1}{2}\frac{y^{'''}}{y^{'}},
\]
we conclude that
\[
V_I(x)=\frac{3}{16 g}\left(\frac{dg}{dy}\right)^2-\frac{1}{4}\frac{d^2g}{dy^2}-\kappa\frac{y^2+(\rho+\rho_1)y+\sigma+\sigma_1}
{R(y)}.
\]
with $y=y(x)$ and $g$ given by (\ref{f}). If we replace (\ref{f}) in the above expression, we obtain as in \cite{Iwata} 
\begin{equation}\label{potgen}
V_I(x)=\frac{\kappa y^2(1-y)^2}{R(y)}\left[\frac{1}{y^2(1-y)^2}+\frac{2}{R(y)}+\frac{(1-2y)(2y+\rho)}{y(1-y)R(y)}
-\frac{5(2y+\rho)^2}{4R^2(y)}\right] -\kappa\frac{y^2+(\rho+\rho_1)y+\sigma+\sigma_1}{R(y)}.
\end{equation}
We call $V_I$ Iwata's potential in order to distinguish it from a similar result obtained in \cite{Natanzon}. Note that 
(\ref{potgen}) reduces to the modified P\"{o}schl-Teller potential barrier \cite{Fl,Landau}
\begin{equation}\label{MPT}
V(x)=\frac{V_0}{\cosh^2{(\alpha x)}},\quad \alpha>0
\end{equation}
whenever
\[
\rho=-1,\quad \sigma=0,\quad \kappa=\alpha^2\quad \rho_1=-\frac{V_0}{\alpha^2},\quad\sigma_1=-\rho_1+\frac{3}{4}. 
\]
The potential $V_I$ will vanish identically whenever 
\[
\rho=-2,\quad \sigma=1,\quad \rho_1=\sigma_1=0, 
\]
\[
\rho=-1,\quad \sigma=0,\quad \rho_1=0,\quad \sigma_1=\frac{3}{4},
\]
or
\[
\rho=\sigma=\rho_1=\sigma_1=0.
\]
Relations connecting the parameters of the hypergeometric equation with those appearing in (\ref{potgen}) can be derived as 
follows. Since the energy of the particle must be real, from (\ref{zero}) we have $(a-b)^2<0$ if $\kappa>0$, and $(a-b)^2>0$ 
if $\kappa<0$. Thus, for $\kappa>0$ we obtain
\begin{equation}\label{c1}
a-b=i\frac{k}{\sqrt{\kappa}} 
\end{equation}
which implies that the real parts of $a$ and $b$ coincide. Moreover, for $\kappa<0$ we get
\[
a-b= \frac{k}{\sqrt{|\kappa|}},
\]
implying that the imaginary parts of $a$ and $b$ coincide. From (\ref{ABC}) and (\ref{stellad}) we have
\begin{eqnarray}
2[c(a+b-1)-2ab]&=&\rho[1-(a-b)^2]+\rho_1\label{tris},\\
2c-c^2&=&\sigma[1-(a-b)^2]+\sigma_1\label{quis}. 
\end{eqnarray}
Replacing (\ref{zero}) into (\ref{quis}) we get
\begin{equation}\label{c}
c=1+\sqrt{1-\sigma-\sigma_1-\sigma\frac{k^2}{\kappa}}.
\end{equation}
Taking into account that (\ref{tris}) can be manipulated as in \cite{Iwata} to give
\[
(1+\sigma+\rho)(a-b)^2-(a+b-c)^2=\rho+\rho_1+\sigma+\sigma_1, 
\]
we find that
\begin{eqnarray}
a&=&\frac{1}{2}\left[1+\sqrt{1-\sigma-\sigma_1-\sigma\frac{k^2}{\kappa}}+i\left(\sqrt{\rho+\rho_1+\sigma+\sigma_1+
\frac{1+\rho+\sigma}{\kappa}k^2}+\frac{k}{\sqrt{\kappa}}\right)\right],\label{a1}\\
b&=&\frac{1}{2}\left[1+\sqrt{1-\sigma-\sigma_1-\sigma\frac{k^2}{\kappa}}+i\left(\sqrt{\rho+\rho_1+\sigma+\sigma_1+
\frac{1+\rho+\sigma}{\kappa}k^2}-\frac{k}{\sqrt{\kappa}}\right)\right],\label{b1} 
\end{eqnarray}
if $\kappa>0$. In the case $\kappa<0$ we get
\begin{eqnarray}
a&=&\frac{1}{2}\left(1+\frac{k}{\sqrt{|\kappa|}}+\sqrt{1-\sigma-\sigma_1-\sigma\frac{k^2}{\kappa}}+i\sqrt{\rho+\rho_1+\sigma+\sigma_1+
\frac{1+\rho+\sigma}{\kappa}k^2}\right),\label{a2}\\
b&=&\frac{1}{2}\left(1-\frac{k}{\sqrt{|\kappa|}}+\sqrt{1-\sigma-\sigma_1-\sigma\frac{k^2}{\kappa}}+i\sqrt{\rho+\rho_1+\sigma+\sigma_1+
\frac{1+\rho+\sigma}{\kappa}k^2}\right)\label{b2}. 
\end{eqnarray}  
We point out that \cite{Iwata} treated only the case when $Y^{'}$ and $Z^{'}$ are linearly independent. In what follows 
we extend the work of \cite{Iwata} by considering two additional cases generating new classes of potentials. To this purpose, 
let $X^{'}$ and $Y^{'}$ be linearly independent. Then, there exist constants $\widetilde{\rho}$ and $\widetilde{\sigma}$ 
such that $Z^{'}+\widetilde{\rho} X^{'}+\widetilde{\sigma} Y^{'}=0$. This condition together with (\ref{tilde}) implies that 
$\dot{A}-\widetilde{\rho}\dot{C}=0$ and $\dot{B}-\widetilde{\sigma}\dot{C}=0$. Integrating these equations with respect to 
$k$ we get $A=\widetilde{\rho} C+\rho_2$ and $B=\widetilde{\sigma} C+\sigma_2$ with $\rho_2,\sigma_2\in\mathbb{R}$. At this 
point $U$ can be expressed as follows
\begin{equation}\label{puntoesclamativo1}
U(x)=T(x)+C\left[\widetilde{\rho}X(x)+\widetilde{\sigma} Y(x)+Z(x)\right]+\rho_2 X(x)+\sigma_2 Y(x).
\end{equation}
If we integrate (\ref{duepunti}) with respect to $x$ and make use of (\ref{XYZ}), we discover that the coordinate 
transformation $y$ must satisfy the following first order nonlinear autonomous differential equation, namely
\begin{equation}\label{eq_y1}
\frac{(y^{'})^2 \widetilde{R}(y)}{4y^2 (1-y)^2}=\widetilde{\kappa}, 
\quad\widetilde{\kappa}\in\mathbb{R}\backslash\{0\},\quad \widetilde{R}(y)=\widetilde{\rho}y^2+\widetilde{\sigma} y+1.
\end{equation}
Taking into account that $U=k^2-V_{II}$ and applying (\ref{puntoesclamativo}), we find that
\[
V_{II}(x)=k^2+\left(\frac{y^{''}}{2y^{'}}\right)^2-\left(\frac{y^{''}}{2y^{'}}\right)^{'}
-C\left[\widetilde{\rho}X(x)+\widetilde{\sigma} Y(x)+Z(x)\right]-\rho_2 X(x)-\sigma_2 Y(x).
\]
With the help of (\ref{XYZ}) and (\ref{eq_y1}) it is straightforward to check that
\[
C\left[\widetilde{\rho}X(x)+\widetilde{\sigma} Y(x)+Z(x)\right]+\rho_2 X(x)+\sigma_2 Y(x)=
\widetilde{\kappa}(2c-c^2)+\frac{\widetilde{\kappa}(\rho_2 y^2+\sigma_2 y)}{\widetilde{R}(y)}.
\]
Since the potential cannot depend on $k$, the parameter $c$ must satisfy the following condition
\begin{equation}\label{zero1}
2c-c^2=\frac{k^2}{\widetilde{\kappa}}
\end{equation}
and the expression for the potential reduces to
\[
V_{II}(x)= \left(\frac{y^{''}}{2y^{'}}\right)^2-\left(\frac{y^{''}}{2y^{'}}\right)^{'}
-\frac{\widetilde{\kappa}(\rho_2 y^2+\sigma_2 y)}{\widetilde{R}(y)}.
\]
The first two terms on the r.h.s. of the above expression can be computed with the same procedure we adopted 
for the case when $Y^{'}$ and $Z^{'}$ are linearly independent. Hence, we end up with a new potential   
\begin{equation}\label{potgen1}
V_{II}(x)=\frac{\widetilde{\kappa} y^2(1-y)^2}{\widetilde{R}(y)}\left[\frac{1}{y^2(1-y)^2}+
\frac{2\widetilde{\rho}}{\widetilde{R}(y)}+\frac{(1-2y)(2\widetilde{\rho}y+\widetilde{\sigma})}{y(1-y)\widetilde{R}(y)}
-\frac{5(2\widetilde{\rho}y+\widetilde{\sigma})^2}{4\widetilde{R}^2(y)}\right] 
-\widetilde{\kappa}\frac{\rho_2 y^2+\sigma_2 y}{\widetilde{R}(y)}.
\end{equation}
The parameters $a$, $b$, and $c$ can be expressed in terms of the parameters of the potential and $k^2$ as follows
\[
a=\frac{c+\Omega}{2}+\frac{1}{2}\sqrt{1-\rho_2-\sigma_2-(1+\widetilde{\rho}+\widetilde{\sigma})\frac{k^2}{\widetilde{k}}},\quad
b=a-\Omega,\quad c=1+\sqrt{1-\frac{k^2}{\widetilde{\kappa}}},\quad
\Omega=\sqrt{1-\rho_2-\widetilde{\rho}\frac{k^2}{\widetilde{\kappa}}}.
\]
It is interesting to observe that the potentials $V_I$ and $V_{II}$ will coincide if and only if
\begin{equation}\label{condoz1}
\widetilde{\rho}=1=\rho_2,\quad\widetilde{\sigma}=\rho,\quad\widetilde{\kappa}=\kappa,\quad
\sigma_2=\rho+\rho_1,\quad\sigma=1,\quad\sigma_1=-1.  
\end{equation}
Moreover, taking into account that in the case of the modified P\"{o}schl-Teller potential the coordinate transformation $y$ 
is given by $y(x)=\cosh^2{(\alpha x)}$ with $\alpha>0$ \cite{Fl}, it is not difficult to verify with the help of 
(\ref{eq_y1}) that this potential does not belong to the class $V_{II}$. A further class of potentials, here denoted as 
$V_{III}$ can be obtained by considering the case when $X^{'}$ and $Z^{'}$ are linearly independent. Then, there exist 
constants $\widehat{\rho}$ and $\widehat{\sigma}$ such that $Y^{'}+\widehat{\rho} X^{'}+\widehat{\sigma} Z^{'}=0$. 
If we use this condition together with (\ref{tilde}), we find that $\dot{A}-\widehat{\rho}\dot{B}=0$ and 
$\dot{C}-\widehat{\sigma}\dot{B}=0$. After integration of these equations with respect to 
$k$ we get $A=\widehat{\rho} B+\rho_3$ and $C=\widehat{\sigma} B+\sigma_3$ with $\rho_3,\sigma_3\in\mathbb{R}$. Also in this  
case $U$ can be expressed as 
\begin{equation}\label{puntoesclamativo2}
U(x)=T(x)+B\left[\widehat{\rho}X(x)+Y(x)+\widehat{\sigma} Z(x)\right]+\rho_3 X(x)+\sigma_3 Z(x).
\end{equation}
If we integrate (\ref{duepunti}) with respect to $x$ and make use of (\ref{XYZ}), we discover that the coordinate 
transformation $y$ must satisfy the following first order nonlinear autonomous differential equation, namely
\begin{equation}\label{eq_y2}
\frac{(y^{'})^2 \widehat{R}(y)}{4y^2 (1-y)^2}=\widehat{\kappa}, 
\quad\widehat{\kappa}\in\mathbb{R}\backslash\{0\},\quad \widehat{R}(y)=\widehat{\rho}y^2+y+\widehat{\sigma}.
\end{equation}
Since $U=k^2-V_{III}$ and applying (\ref{puntoesclamativo}), we find that
\[
V_{II}(x)=k^2+\left(\frac{y^{''}}{2y^{'}}\right)^2-\left(\frac{y^{''}}{2y^{'}}\right)^{'}
-B\left[\widehat{\rho}X(x)+Y(x)+\widehat{\sigma}Z(x)\right]-\rho_3 X(x)-\sigma_3 Z(x).
\]
With the help of (\ref{XYZ}) and (\ref{eq_y2}) it can be checked that
\[
B\left[\widehat{\rho}X(x)+Y(x)+\widehat{\sigma}Z(x)\right]+\rho_3 X(x)+\sigma_3 Z(x)=
2\widehat{\kappa}\left[c(a+b-1)-2ab\right]+\frac{\widehat{\kappa}(\rho_3 y^2+\sigma_3)}{\widehat{R}(y)}.
\]
If we choose the parameters $a$, $b$ and $c$ so that 
\begin{equation}\label{zero2}
c(a+b-1)-2ab=\frac{k^2}{2\widehat{\kappa}},
\end{equation}
we ensure that the potential $V_{III}$ will not depend on the energy of the particle. Hence, we obtain
\[
V_{III}(x)= \left(\frac{y^{''}}{2y^{'}}\right)^2-\left(\frac{y^{''}}{2y^{'}}\right)^{'}
-\frac{\widehat{\kappa}(\rho_3 y^2+\sigma_3)}{\widehat{R}(y)}.
\]
Computing the first two terms on the r.h.s. of the above expression as we did in the previous cases, we end up with the 
following potential   
\begin{equation}\label{potgen2}
V_{III}(x)=\frac{\widehat{\kappa} y^2(1-y)^2}{\widehat{R}(y)}\left[\frac{1}{y^2(1-y)^2}+
\frac{2\widehat{\rho}}{\widehat{R}(y)}+\frac{(1-2y)(2\widehat{\rho}y+1)}{y(1-y)\widehat{R}(y)}
-\frac{5(2\widehat{\rho}y+1)^2}{4\widehat{R}^2(y)}\right] 
-\widehat{\kappa}\frac{\rho_3 y^2+\sigma_3}{\widehat{R}(y)}.
\end{equation}
The parameters $a$, $b$, and $c$ can be expressed in terms of the parameters of the potential and $k^2$ as follows
\[
a=\frac{c+\widetilde{\Omega}}{2}+\frac{1}{2}\sqrt{1-\rho_3-\sigma_3-(1+\widehat{\rho}+\widehat{\sigma})\frac{k^2}{\widehat{k}}},\quad
b=a-\widetilde{\Omega},\quad c=1+\sqrt{1-\sigma_3-\widehat{\sigma}\frac{k^2}{\widetilde{\kappa}}},\quad
\widetilde{\Omega}=\sqrt{1-\rho_3-\widehat{\rho}\frac{k^2}{\widehat{\kappa}}}.
\]
Already comparing the last two terms in (\ref{potgen1}) and (\ref{potgen2}) we see that in general the potentials 
$V_{II}$ and $V_{III}$ do not in general coincide. They will coincide if and only if
\begin{equation}\label{condo2}
\widehat{\kappa}=\widetilde{\kappa},\quad
\widehat{\rho}=\widetilde{\rho},\quad
\widetilde{\sigma}=1=\widehat{\sigma},\quad
\sigma_2=0=\sigma_3,\quad
\rho_2=\rho_3.  
\end{equation}
Moreover, the potentials $V_I$ and $V_{III}$ will coincide whenever 
\begin{equation}\label{condo3}
\widehat{\kappa}=\kappa,\quad
\widehat{\rho}=\rho=\rho_3=1,\quad
\rho_1=-1,\quad
\widehat{\sigma}=\sigma,\quad
\sigma+\sigma_1=\sigma_3. 
\end{equation}
Finally, by means of (\ref{condoz1}), (\ref{condo2}) and (\ref{condo3}) it can be easily verified that there is only one potential living in all 
three classes. This happens when
\[
\kappa=\widetilde{\kappa}=\widehat{\kappa},\quad
\rho=\widetilde{\rho}=\widehat{\rho}=\rho_2=\rho_3=1,\quad
\rho_1=\sigma_1=-1,\quad
\sigma=\widetilde{\sigma}=\widehat{\sigma}=1,\quad
\sigma_2=\sigma_3=0. 
\]
Trivially, Iwata's potential becomes an even function whenever the coordinate transformation $y$ is even. Is it possible for $V_I$ to be 
symmetric even though the coordinate transformation is not? The next result shows that if $y$ is nonsymmetric and a certain set of 
functions is linearly independent, the potential must be zero, i.e. the only symmetric potential for such a choice of $y$ is the 
trivial potential. 
\begin{theorem}\label{teorema_1}
Let $D\subseteq\mathbb{R}$ be a common domain of definition for $y$ and $v$ with $v(x)=y(-x)$ and suppose that $y(x)\neq v(x)$ for 
at least one $x\in D\subseteq\mathbb{R}$. If the functions in the set  $\mathcal{U}=\bigcup_{i=1}^3\mathcal{U}_i$ with $\mathcal{U}_1=\{y^n-v^n~|~n=1,\cdots,6\}$, $\mathcal{U}_2=\{v^6 y^n-y^6 v^n~|~n=1,\cdots,5\}$, 
and $\mathcal{U}_3=\{y^m v^n-v^m y^n~|~m,n=1,\cdots,5;~m>n\}$ are linearly independent, then the potential $V_I$ 
as given in (\ref{potgen}) vanishes under the assumption that it is symmetric, i.e. $V_I(x)=V_I(-x)$.
\end{theorem}
\begin{proof}
If we rewrite the potential as
\[
V_I(x)=-\frac{p(y)}{4q(y)},\quad p(y)=\sum_{n=0}^{5} a_n y^n,\quad q(y)=R^3(y)=\sum_{n=0}^{6} b_n y^n
\]
with
\begin{eqnarray*}
a_0&=&4\sigma^2(\sigma+\sigma_1-1),\quad a_1=4\sigma(\rho_1\sigma+3\rho\sigma+2\sigma_1\rho-3\rho),\\
a_2&=&12\rho\sigma+12\sigma^2+8\rho\rho_1\sigma-3\rho^2+4\sigma_1\rho^2+12\rho^2\sigma-24\sigma+8\sigma_1\sigma,\\
a_3&=&16\rho\sigma+40\sigma+8\sigma_1\rho+4\rho_1\rho^2+8\rho_1\sigma+4\rho^3+2\rho^2-8\rho,\\
a_4&=&9\rho^2+4\sigma_1-12\sigma+8\rho_1\rho+12\rho,\quad a_5=4\rho_1,\\
b_0&=&\sigma^3,\quad b_1=3\rho\sigma^2,\quad b_2=3\rho^2\sigma+3\sigma^2,\quad b_3=\rho^3+6\rho\sigma,\quad
b_4=3\sigma+3\rho^2,\\
b_5&=&3\rho,\quad b_6=1
\end{eqnarray*}
and set $v=y(-x)$, then $V_I$ will be symmetric, i.e. $V_I(x)=V_I(-x)$ if and only if
\[
p(y)q(v)-p(v)q(y)=0. 
\]
This relation can be equivalently written as
\begin{equation}\label{preludio}
\sum_{n=0}^{5} (b_0 a_n-a_0 b_n) (y^n-v^n)-a_0(y^6-v^6)+\sum_{n=1}^{5} a_n(v^6 y^n-y^6 v^n)+
\sum_{m,n=1}^{5} a_m b_n \varphi_{mn}=0,\quad\varphi_{mn}=y^m v^n-v^m y^n,
\end{equation}
where we made use of the fact that $b_6=1$. Since the last term in the above equation vanishes for $n=m$ and 
$\varphi_{mn}=-\varphi_{nm}$, we finally obtain the following more economic version of (\ref{preludio})
\begin{equation}\label{indep}
\sum_{n=0}^{5} (b_0 a_n-a_0 b_n) (y^n-v^n)-a_0(y^6-v^6)+\sum_{n=1}^{5} a_n(v^6 y^n-y^6 v^n)+\sum_{\substack{
   m,n=1 \\
   m>n
  }}^5
  a_m b_n \varphi_{mn}=0.
\end{equation}
Let us introduce the sets of functions $\mathcal{U}_1=\{y^n-v^n~|~n=1,\cdots,6\}$, $\mathcal{U}_2=\{v^6 y^n-y^6 v^n~|~n=1,\cdots,5\}$, 
and $\mathcal{U}_3=\{y^m v^n-v^m y^n~|~m,n=1,\cdots,5;~m>n\}$. If the functions in the set $\mathcal{U}=\bigcup_{i=1}^3\mathcal{U}_i$ 
are linearly independent, we can conclude that (\ref{indep}) will hold true for every $x$ whenever all coefficients in (\ref{indep}) 
vanish, i.e.
\begin{equation}\label{fine}
a_0=0,\quad b_0 a_n-a_0 b_n=0~\forall~n=1,\cdots,5,\quad a_1=\cdots=a_5=0. 
\end{equation}
Note that the last condition implies that $a_m b_n=0$ for any $m,n=1,\cdots,5$. Since $a_0=a_1=\cdots=a_5=0$, then $p(y)=0$ and thus
$V_I$ must vanish.~~$\square$
\end{proof}
\begin{remark}
It is obvious that (\ref{eq_y}) has a stationary solution $y_1=0$. Hence, among all possible nonsymmetric coordinate transformations those such 
that $y(-x)=-y(x)$ must be ruled out since they would intersect the solution $y_1$ and this would violate the Existence and Uniqueness 
Theorem for first order differential equations. 
\end{remark}
The system (\ref{fine}) has been solved with Maple $15$ and as a result we found that the only solution is given by 
$\rho=\sigma=\rho_1=\sigma_1=0$. For this set of values of the parameters equation (\ref{eq_y}) admits the following explicit 
solutions
\[
y_{\pm}(x)=1+C_{\pm}e^{\pm 2\sqrt{\kappa}x}. 
\]
Let $Y$ denote the set of all nonsymmetric solutions of (\ref{eq_y}) such that the set of functions 
$\mathcal{U}$ introduced in Theorem~\ref{teorema_1} is linearly independent and $S$ be the set of those elements in $Y$ such that the 
potential is symmetric. The set $S$ cannot be empty because if we take for instance $\kappa=1/4$, $\rho=\sigma=0$, $\rho_1=-\sigma_1$, 
and $y(x)=1+e^x$, then we obtain the following symmetric potential, namely
\[
V_I(x)=\frac{\sigma_1 e^x}{(1+e^x)^2}. 
\]
This shows that in general the potential $V_I$ can be symmetric even though the coordinate transformation is not.

\section{Natanzon's potential}\label{nata}
In the same year as Iwata, Natanzon derived the most general potential such that the ODSE can be transformed into the hypergeometric 
equation. The starting point in his derivation is the following relation linking the potential $V$ entering in the ODSE with the 
Bose invariant $I$ \cite{Bose} associated to the hypergeometric equation and the Schwarzian derivative $\{y,x\}$ of the 
coordinate transformation \cite{Natanzon}, namely
\begin{equation}\label{1n}
(y^{'})^2 I(y)+\frac{1}{2}\{y,x\}=k^2-V(x), 
\end{equation}
with
\begin{equation}\label{2n}
I(y)=\frac{(1-\widetilde{\lambda}_0^2)(1-y)+(1-\widetilde{\lambda}_1^2)y+(\mu^2-1)y(1-y)}{4y^2(1-y)^2},\quad
\{y,x\}=\frac{y^{''}}{y^{'}}\left(\frac{y^{'''}}{y^{''}}-\frac{3}{2}\frac{y^{''}}{y^{'}}\right).
\end{equation}
Here, the parameters $\mu$, $\widetilde{\lambda}_0$ and $\widetilde{\lambda}_1$ are related to the parameters $\alpha$, $\beta$, and $\gamma$ of the 
hypergeometric equation as follows
\begin{equation}\label{mn}
\widetilde{\lambda}_0=\gamma-1,\quad \widetilde{\lambda}_1=\alpha+\beta-\gamma,\quad\mu=\beta-\alpha. 
\end{equation}
Under the assumption that the coordinate transformation $y$ does not depend on the energy of the particle  and that the parameters 
$\mu$, $\widetilde{\lambda}_0$ and $\widetilde{\lambda}_1$ are linear in $k^2$, i.e.
\begin{equation}\label{4n}
1-\mu^2=ak^2-f,\quad 1-\widetilde{\lambda}^2_p=c_p k^2-h_p\quad p=0,1 
\end{equation}
Natanzon derived the following result \cite{Natanzon}.
\begin{theorem}
The most general potential such that the ODSE can be transformed into the 
hypergeometric equation is given by
\begin{equation}\label{nat_pot} 
V(x)=\frac{y^2(1-y)^2}{H^2(y)}\left[a+\frac{a+(c_1-c_0)(2y-1)}{y(y-1)}-\frac{5}{4}\frac{\Delta}{H(y)}\right]+
\frac{fy(y-1)+h_0(1-y)+h_1 y+1}{H(y)}
\end{equation}
with $\Delta=(a-c_0-c_1)^2-4c_1 c_0$ and $H(y)=ay^2+(c_1-c_0-a)y+c_0$. Moreover, the coordinate transformation $y=y(x)$ must satisfy the 
differential equation
\[
\frac{(y^{'})^2 H(y)}{4y^2(1-y)^2}=1.
\]
\end{theorem}
\begin{proof} 		 
Let $y=y(x)$ be the coordinate transformation leading from the hypergeometric equation to the ODSE. Such a transformation is 
determined by the condition (\ref{1n}). By means of (\ref{1n}) and (\ref{2n}) we obtain 
\begin{equation}\label{5n}
V(x)=k^2-\frac{(y^{'})^2\{k^2[ay^2+(c_1-c_0-a)y+c_0]-[fy^2+(h_1-h_0-f)y+h_0]\}}{4y^2(1-y)^2}-\frac{1}{2}\{y,x\}.
\end{equation}
Assuming that (\ref{4n}) holds and that $y$ does not depend on the energy of the particle we end up with the following expression for the 
potential  
\[
V(x)=k^2-\frac{(y^{'})^2}{4y^2(1-y)^2}\{k^2\left[ay^2+(c_1-c_0-a)y+c_0\right]-
\left[fy^2+(h_1-h_0-f)y+h_0\right]\}-\frac{1}{2}\{y,x\}.
\]
Let $H(y)= ay^2+(c_1-c_0-a)y+c_0$. Then the above expression can be rewritten as
\[
V(x)=k^2-k^2\frac{(y^{'})^2H(y)}{4y^2(1-y)^2}+\frac{(y^{'})^2\left[fy^2+(h_1-h_0-f)y+h_0\right]}{4y^2(1-y)^2}
-\frac{1}{2}\{y,x\}.
\]
The requirement that the potential does not depend on the energy of the particle implies that the coordinate transformation must 
satisfy the differential equation 
\[
\frac{(y^{'})^2H(y)}{4y^2(1-y)^2}=1.
\]
Hence, we obtain the following expression for the potential
\[
V(x)=-\frac{1}{2}\{y,x\}+\frac{fy^2+(h_1-h_0-f)y+h_0}{H(y)}.
\]
Using (\ref{2n}) the above relation can be rewritten as
\begin{equation}\label{U1}
V(x)=\frac{3}{4}\left(\frac{y^{''}}{y^{'}}\right)^2-\frac{y^{'''}}{2y^{'}}+\frac{fy^2+(h_1-h_0-f)y+h_0}{H(y)}.
\end{equation}
Let us set $\widetilde{g}(y)=4y^2(1-y)^2/H(y)$. Then, the ODE for $y$ becomes $(y^{'})^2=\widetilde{g}(y)$ and differentiating it we 
can rewrite the terms entering in the Schwarzian derivative as
\[
\frac{y^{'''}}{y^{'}}=\frac{1}{2}\frac{d^2\widetilde{g}}{dy^2}
\]
and thus
\[
\frac{3}{4}\left(\frac{y^{''}}{y^{'}}\right)^2-\frac{y^{'''}}{2y^{'}}=\frac{3}{16\widetilde{g}}
\left(\frac{d\widetilde{g}}{dy}\right)^2-\frac{1}{4}\frac{d^2\widetilde{g}}{dy^2}.
\]
Hence, (\ref{U1}) can be rewritten as follows
\begin{equation}\label{U2}
V(x)= \frac{3}{16\widetilde{g}}
\left(\frac{d\widetilde{g}}{dy}\right)^2-\frac{1}{4}\frac{d^2\widetilde{g}}{dy^2}+\frac{fy^2+(h_1-h_0-f)y+h_0}{H(y)}.
\end{equation}
Taking into account that
\begin{eqnarray*}
\frac{d\widetilde{g}}{dy}&=&\frac{8y(1-y)(1-2y)}{H(y)}-\frac{4y^2(1-y)^2}{H^2(y)}G(y),\quad
G(y)=2ay+c_1-c_0-a,\\
\left(\frac{d\widetilde{g}}{dy}\right)^2&=&\frac{64 y^2(1-y)^2(1-2y)^2}{H^2(y)}-
\frac{64y^3(1-y)^3(1-2y)}{H^3(y)}G(y)+\frac{16y^4(1-y)^4}{H^4(y)}G^2(y),\\
\frac{3}{16\widetilde{g}}\left(\frac{d\widetilde{g}}{dy}\right)^2&=&\frac{3(1-2y)^2}{H(y)}-
\frac{3y(1-y)(1-2y)}{H^2(y)}G(y)+\frac{3y^2(1-y)^2}{H^3(y)}G^2(y),\\
\frac{d^2\widetilde{g}}{dy^2}&=&\frac{8(6y^2-6y+1)}{H(y)}-\frac{16y(1-y)(1-2y)}{H^2(y)}G(y)
-\frac{8ay^2(1-y)^2}{H^2(y)}+\frac{8y^2(1-y)^2}{H^3(y)}G^2(y).
\end{eqnarray*}
we obtain
\[
\frac{3}{16\widetilde{g}}
\left(\frac{d\widetilde{g}}{dy}\right)^2-\frac{1}{4}\frac{d^2\widetilde{g}}{dy^2}=
\frac{1}{H(y)}+\frac{y(1-y)(1-2y)}{H^2(y)}G(y)+\frac{2a y^2(1-y)^2}{H^2(y)}
-\frac{5y^2(1-y)^2}{4H^3(y)}G^2(y).
\]
Finally, substitution of the above expression into (\ref{U2}) gives
\begin{equation}\label{corr_Natanzon}
V(x)=\frac{y^2(1-y)^2}{H^2(y)}\left[2a+\frac{H(y)}{y^2(1-y)^2}+\frac{(1-2y)G(y)}{y(1-y)}
-\frac{5}{4}\frac{G^2(y)}{H(y)}\right]+\frac{fy^2+(h_1-h_0-f)y+h_0}{H(y)}.
\end{equation}
By rewriting  $G^2(y)$ in terms of $\Delta=(a-c_0-c_1)^2-4c_1 c_0$ some trivial algebra leads to (\ref{nat_pot}).~~$\square$
\end{proof}
The next result shows the Natanzon class of potentials contains all classed of Iwata's potentials. This indicates that the method 
employed by Iwata is less general than that used by Natanzon.
\begin{theorem}
Iwata's potentials (\ref{potgen}), (\ref{potgen1}), and (\ref{potgen2}) are a subclass of the Natanzon class (\ref{nat_pot}).
\end{theorem}
\begin{proof}
In Section~\ref{iwa} we derived formulae expressing the parameters $a$, $b$, and $c$ in terms of the parameters entering in Iwata's 
potentials (\ref{potgen}), (\ref{potgen1}), and (\ref{potgen2}). Substitution of these formulae into (\ref{4n}) allows to express the 
parameters of Natanzon's potential in terms of those for the Iwata potentials. Hence, (\ref{nat_pot}) reproduces 
the potential $V_I$ in (\ref{potgen}) for
\begin{equation}\label{cc1}
c_0=\frac{\sigma}{\kappa},\quad h_0=-(\sigma+\sigma_1),\quad
c_1=\frac{1+\rho+\sigma}{\kappa},\quad h_1=-(1+\rho+\rho_1+\sigma+\sigma_1),\quad a=\frac{1}{\kappa},\quad f=-1, 
\end{equation}
the potential $V_{II}$ in (\ref{potgen1}) if
\begin{equation}\label{cc2}
c_0=\frac{1}{\widetilde{\kappa}},\quad h_0=0,\quad
c_1=\frac{1+\widetilde{\rho}+\widetilde{\sigma}}{\widetilde{\kappa}},\quad h_1=-(\rho_2+\sigma_2),\quad 
a=\frac{\widetilde{\rho}}{\widetilde{\kappa}},\quad f=-\rho_2, 
\end{equation}
and the potential $V_{III}$ in (\ref{potgen2}) whenever
\begin{equation}\label{cc3}
c_0=\frac{\widehat{\sigma}}{\widehat{\kappa}},\quad h_0=-\sigma_3,\quad
c_1=\frac{1+\widehat{\rho}+\widehat{\sigma}}{\widehat{\kappa}},\quad h_1=-(\rho_3+\sigma_3),\quad 
a=\frac{\widehat{\rho}}{\widehat{\kappa}},\quad f=-\rho_3. 
\end{equation}
At this point it is not difficult to construct a potential in the Natanzon class which does not belong to any Iwata class. For 
example, we can start by choosing choose $f=0$, and $h_0=1$. This ensures that the potential cannot belong to the first or second 
Iwata class. Since we have $h_1=-\rho_3-\sigma_3=f+h_0$ in (\ref{cc3}), we can choose $h_1=0$, thus ensuring that this potential cannot be 
reproduced by any of the (\ref{potgen}), (\ref{potgen1}), and (\ref{potgen2}). Hence, Iwata's classes are contained in Natanzon's 
class.~~$\square$
\end{proof}.

\section{Generalized Natanzon potentials of the Heun class}\label{heunclass}
We shortly review some aspects of Milson's method \cite{Milson} playing a fundamental role in the derivation of the most general 
potentials such that the ODSE can be reduced to a Heun equation or one of its confluent cases. Let us consider the second order 
linear equation
\begin{equation}\label{1}
a(y)\frac{d^2 v}{dy^2}+b(y)\frac{dv}{dy}+c(y)v(y)=0. 
\end{equation}
By means of the Liouville transformation
\[
v(y)={\rm{exp}}\left(-\int_y\frac{b(t)}{2a(t)}~dt\right) f(y) 
\]
we find that the standard form of (\ref{1}) is
\begin{equation}\label{2}
\frac{d^2 f}{dy^2}+I(y)f(y)=0,\quad I(y)=\frac{1}{4a^2(y)}\left[4a(y)c(y)-2a(y)\frac{db}{dy}+2b(y)\frac{da}{dy}-b^2(y)\right], 
\end{equation}
where $I$ is the Bose invariant of (\ref{1})\cite{Bose,Milson}. If we make the coordinate transformation $y=y(x)$ and set 
$\widetilde{f}(x)=f(y(x))$, it can be seen that (\ref{2}) becomes
\begin{equation}\label{3}
 \frac{1}{(y^{'})^2}\frac{d^2\widetilde{f}}{dx^2}-\frac{y^{''}}{(y^{'})^2}\frac{d\widetilde{f}}{dx}+I(y)\widetilde{f}(x)=0,
\end{equation}
where the prime denotes differentiation with respect to $x$. The standard form of (\ref{3}) can be attained by another Liouville 
transformation, namely
\[
\widetilde{f}(x)= {\rm{exp}}\left(-\int_x\frac{y^{''}}{2y^{'}}~ds\right)u(x)=\sqrt{y^{'}}u(x) 
\]
and we obtain
\begin{equation}\label{4}
\frac{d^2 u}{dx^2}+J(x)u(x)=0,\quad
J(x)=(y^{'})^2 I(y)+\frac{1}{2}\{y,x\},
\end{equation}
where the Schwarzian derivative is given by
\[
\{y,x\}=\left(\frac{y^{''}}{y^{'}}\right)^{'}-\frac{1}{2}\left(\frac{y^{''}}{y^{'}}\right)^2
=\frac{y^{''}}{y^{'}}\left(\frac{y^{'''}}{y^{''}}-\frac{3}{2}\frac{y^{''}}{y^{'}}\right). 
\]
The solution of (\ref{4}) will be expressed in terms of the solutions of (\ref{1}) by the following relation
\[
u(x)=(y^{'})^{-1/2} {\rm{exp}}\left(\int_y\frac{b(t)}{2a(t)}~dt\right)v(y(x)).
\]
A comparison of the above expression with (\ref{treI}) show that the function $s$ introduced by Iwata will be given by the following 
formula
\[
s(x)= \sqrt{y^{'}}{\rm{exp}} \left(-\int_y\frac{b(t)}{2a(t)}~dt\right).
\]
Equation (\ref{4}) will reduce to an ODSE whenever $J(x)=k^2-V(x)$. Hence, the potential $V$ is entirely specified by the Bose 
invariant and the Schwarzian derivative of the coordinate transformation $y$. According to \cite{Milson} the potential $V$ of the 
ODSE will not depend on $k^2$ if the Bose invariant admits a decomposition of the form $I(y)=I_1(y)k^2+I_0(y)$ and the coordinate 
transformation $y$ is a solution of the autonomous differential equation
\begin{equation}\label{ODE_y}
 (y^{'})^2=\frac{1}{I_1(y)}.
\end{equation}
If this is the case, then the most general potential such that the ODSE can be transformed into (\ref{1}) is
\begin{equation}\label{Milson_pot}
V(x)=-\frac{I_0(y)}{I_1(y)}+\frac{4I_1(y)\ddot{I}_1(y)-5(\dot{I}_1)^2}{16I_1^3(y)},
\end{equation}
where a dot represents differentiation with respect to $y$. In what follows we extend the work of \cite{Milson} by deriving expressions for 
the potential $V$ in the case that (\ref{1}) is a Heun equation, a confluent Heun equation, a double confluent equation, a biconfluent 
Heun equation or a triconfluent Heun equation. We underline that Lemma~$3$ in \cite{Milson} cannot be applied for the Heun equation and 
its confluent cases because it would require that in the Bose invariant
\[
I(y)=\frac{T(y,k^2)}{A^2(y,k^2)} 
\]
the polynomial $A$ has has at most degree two, whereas the corresponding $A$ for the Heun equation and its confluent cases has degree 
three. The next result shows that the most general potential associated to the Heun equation is controlled by ten parameters.
\begin{theorem}\label{potenzialeHeun}
Let $f_0,\cdots,f_4,g_0,\cdots,g_4$ be real parameters. The most general potential such that the ODSE can be transformed into the 
Heun equation (\ref{Heun}) is given by
\begin{equation}\label{heun_pot} 
V_H(x)=\frac{y^2(y-1)^2(y-a)^2}{R^2(y)}\left[\ddot{R}(y)+\frac{G(y)\dot{R}(y)-2R(y)\dot{G}(y)}{y(y-1)(y-a)}+
\frac{R(y)G^2(y)}{y^2(y-1)^2(y-a)^2}-\frac{5}{4}\frac{(\dot{R}(y))^2}{R(y)}\right]-\frac{S(y)}{R(y)},
\end{equation}
where $G(y)=3y^2-2(a+1)y+a$ and
\begin{eqnarray*}
R(y)&=&g_0(y-1)^2(y-a)^2+g_1 y^2(y-a)^2+g_2 y^2(y-1)^2+g_3y^2(y-1)(y-a)+g_4y(y-1)(y-a),\\
S(y)&=&f_0(y-1)^2(y-a)^2+f_1 y^2(y-a)^2+f_2 y^2(y-1)^2+f_3y^2(y-1)(y-a)+f_4y(y-1)(y-a).
\end{eqnarray*}
Moreover, the coordinate transformation $y=y(x)$ satisfies the differential equation
\[
(y^{'})^2 =\frac{4y^2(y-1)^2(y-a)^2}{R(y)}.
\]
\end{theorem}
\begin{proof}
By means of the relation $\epsilon=\alpha+\beta+1-\gamma-\delta$ we can cast the Bose invariant of the Heun equation into the form
\begin{equation}\label{BoseHeun}
I(y)=\frac{\lambda_0(y-1)^2(y-a)^2+\lambda_1 y^2(y-a)^2+\lambda_2 y^2(y-1)^2+\lambda_3y^2(y-1)(y-a)+\lambda_4y(y-1)(y-a)}
{4y^2(y-1)^2(y-a)^2}
\end{equation}
with
\[
\lambda_0=1-(1-\gamma)^2,\quad
\lambda_1=1-(1-\delta)^2,\quad
\lambda_2=1-(\alpha+\beta-\gamma-\delta)^2, 
\]
\[
\lambda_3=4\alpha\beta-2\gamma\delta-2(\gamma+\delta)(\alpha+\beta+1-\gamma-\delta),\quad
\lambda_4=-4q+2a\gamma\delta+2\gamma(\alpha+\beta+1-\gamma-\delta). 
\]
Note that the Bose invariant is a rational function of the form
\[
I(y)=\frac{P(y)}{Q(y)}=\frac{c_0+c_1 y+c_2 y^2+c_3 y^3+c_4 y^4}{4y^2(y-1)^2(y-a)^2}
\]
and it admits the decomposition $I(y)=I_1(y)k^2+I_0(y)$ if and only if $P(y)=R(y)k^2+S(y)$ for some polynomials 
$R(y)=a_0+a_1y+a_2y^2+a_3y^3+a_4y^4$ and $S(y)=b_0+b_1y+b_2y^2+b_3y^3+b_4y^4$. This will be the case if
\begin{eqnarray}
c_0&=&a_0k^2+b_0=a^2\lambda_0,\label{1p}\\
c_1&=&a_1k^2+b_1=-2a(a+1)\lambda_0+a\lambda_4,\label{2p}\\
c_2&=&a_2k^2+b_2=(a^2+4a+1)\lambda_0+a^2\lambda_1+\lambda_2+a\lambda_3-(a+1)\lambda_4,\label{3p}\\
c_3&=&a_3k^2+b_3=-2(a+1)\lambda_0-2a\lambda_1-2\lambda_2-(a+1)\lambda_3+\lambda_4,\label{4p}\\
c_4&=&a_4k^2+b_4=\lambda_0+\lambda_1+\lambda_2+\lambda_3+\lambda_4\label{5p}. 
\end{eqnarray}
Note that (\ref{1p}) implies that $\lambda_0$ is linear in $k^2$ and from (\ref{2p}) we conclude that $\lambda_4$ is also linear 
in $k^2$. We show now that $\lambda_1$, $\lambda_2$ and $\lambda_3$ are linear in $k^2$. From (\ref{3p})-(\ref{5p}) we see that 
$a^2\lambda_1+\lambda_2+a\lambda_3$, $-2a\lambda_1-2\lambda_2-(a+1)\lambda_3$, and $\lambda_1+\lambda_2+\lambda_3$ are also linear 
in $k^2$. Let
\begin{eqnarray*}
\lambda_1+\lambda_2+\lambda_3&=&r_1k^2+s_1,\\
-2a\lambda_1-2\lambda_2-(a+1)\lambda_3&=&r_2k^2+s_2,\\
a^2\lambda_1+\lambda_2+a\lambda_3&=&r_3k^2+s_3.
\end{eqnarray*}
for some scalars $r_1,\cdots,r_3,s_1,\cdots,s_3$. Let us rewrite the above non homogeneous system in matrix form, i.e.
\[
A\Lambda=k^2 T+B,\quad 
A= \left( \begin{array}{ccc}
1 & 1 & 1 \\
-2a & -1 & -(a+1) \\
a^2 & 1 & a \end{array} \right),\quad
\Lambda=\left( \begin{array}{c}
\lambda_1  \\
\lambda_2 \\
\lambda_3  \end{array} \right),\quad 
T=\left( \begin{array}{c}
r_1  \\
r_2 \\
r_3  \end{array} \right),\quad
S=\left( \begin{array}{c}
s_1  \\
s_2 \\
s_3  \end{array} \right).
\]
Since the parameter $a$ in the Heun equation cannot assume the values $0$ and $1$, we conclude that ${\rm{det}}(A)=(1-a)^3\neq 0$. 
Hence, the matrix $A$ is invertible and therefore $\lambda_1,\lambda_2,\lambda_3$ are linear in $k^2$. By setting 
$\lambda_i=g_ik^2+f_i$ with $i=0,1,\cdots,4$ we can rewrite the Bose invariant (\ref{BoseHeun}) in the form $I(y)=I_1(y)k^2+I_0(y)$ 
with
\[
I_1(y)=\frac{R(y)}{4y^2(y-1)^2(y-a)^2},\quad  I_0(y)=\frac{S(y)}{4y^2(y-1)^2(y-a)^2},
\]
and
\begin{eqnarray*}
R(y)&=&g_0(y-1)^2(y-a)^2+g_1 y^2(y-a)^2+g_2 y^2(y-1)^2+g_3y^2(y-1)(y-a)+g_4y(y-1)(y-a),\\
S(y)&=&f_0(y-1)^2(y-a)^2+f_1 y^2(y-a)^2+f_2 y^2(y-1)^2+f_3y^2(y-1)(y-a)+f_4y(y-1)(y-a).
\end{eqnarray*}
Taking into account that
\[
\dot{I}_1(y)=\frac{\dot{R}(y)}{4y^2(y-1)^2(y-a)^2}-\frac{R(y)G(y)}{2y^3(y-1)^3(y-a)^3},\quad G(y)=3y^2-2(a+1)y+a,
\]
we find that
\begin{equation}\label{I11}
\frac{(\dot{I}_1)^2}{I_1}= \frac{(\dot{R}(y))^2}{4y^2(y-1)^2(y-a)^2 R(y)}+
\frac{R(y)G^2(y)}{y^4(y-1)^4(y-a)^4}-\frac{\dot{R}(y)G(y)}{y^3(y-1)^3(y-a)^3}.
\end{equation}
Moreover,
\begin{equation}\label{I22}
\ddot{I}_1(y)=\frac{\ddot{R}(y)}{4y^2(y-1)^2(y-a)^2}-\frac{\dot{R}(y)G(y)+R(y)[3y-(a+1)]}{y^3(y-1)^3(y-a)^3} 
+\frac{3}{2}\frac{R(y)G^2(y)}{y^4(y-1)^4(y-a)^4}.
\end{equation}
Finally, substituting (\ref{I11}) and (\ref{I22}) into (\ref{Milson_pot}) we obtain (\ref{heun_pot}).~~$\square$
\end{proof}
In the next result we show that the Natanzon class is a special case of the potential (\ref{heun_pot}).
\begin{corollary}
Natanzon's class of solvable potentials (\ref{nat_pot}) is a subclass of the family of potentials (\ref{heun_pot}). 
\end{corollary}
\begin{proof}
The Bose invariant of the Heun equation will collapse into the Bose invariant of the hypergeometric equation whenever $q=\alpha\beta a$ 
and $\epsilon=\alpha+\beta+1-\gamma-\delta=0$. In this case the parameters $\lambda_0,\cdots,\lambda_4$ in (\ref{BoseHeun}) become 
\[
\lambda_0=1-(1-\gamma)^2,\quad
\lambda_1=1-(\alpha+\beta-\gamma)^2,\quad
\lambda_2=0,\quad
\lambda_3=4\alpha\beta-2\gamma(\alpha+\beta+1-\gamma),\quad
\lambda_4=-a\lambda_3. 
\]
Since $a$ denotes the position of a regular singular point in (\ref{Heun}), we rename the parameters entering in Natanzon's potential 
according to the rule $a\to\widehat{a}$, $c_i\to \widehat{c}_i$, $h_i\to\widehat{h}_i$, and $f\to\widehat{f}$, where $i=0,1$. Then, in the 
notation of the previous theorem and making use of (\ref{mn}) and (\ref{4n}) we find that $g_0k^2+f_0=\lambda_0
=\widetilde{\lambda}_0=\widehat{c}_0 k^2-\widehat{h}_0$. Hence, it must be $g_0=\widehat{c}_0$ and $f_0=-\widehat{h}_0$. Moreover, 
$g_1 k^2+f_1=\lambda_1=\widetilde{\lambda}_1=\widehat{c}_1k^2-\widehat{h}_1$ and therefore we have $g_1=\widehat{c}_1$ and 
$f_1=-\widehat{h}_1$. Furthermore, $\lambda_2=0$ implies that $g_2=f_2=0$. If we rewrite $\lambda_3$ as 
$\lambda_3=[1-(\alpha-\beta)^2]+[1-(\alpha+\beta-\gamma)^2]-[1-(\gamma-1)^2]$ and use (\ref{mn}) and (\ref{4n}) we find that
\[
g_3 k^2+f_3=(\widehat{a}-\widehat{c}_0-\widehat{c}_1)k^2+\widehat{h}_0+\widehat{h}_1-\widehat{f}. 
\]
Hence, $g_3=\widehat{a}-\widehat{c}_0-\widehat{c}_1$ and $f_3=\widehat{h}_0+\widehat{h}_1-\widehat{f}$. Finally, from the relation 
$\lambda_4=-a\lambda_3$ we obtain $g_4=-ag_3$ and $f_4=-af_3$. This completes the proof.~~$\square$
\end{proof}
We consider now the case of the confluent Heun equation \cite{Ronveaux}
\begin{equation}\label{CHE}
\frac{d^2 v}{dy^2}+\left(\frac{\gamma}{y}+\frac{\delta}{y-1}+4p\right)\frac{dv}{dy}+\frac{4p\alpha y-\sigma}{y(y-1)}v(y)=0,
\end{equation}
where $\gamma,\delta,p,\alpha,\sigma$ are arbitrary parameters.
\begin{theorem}
Let $\sigma_0,\cdots,\sigma_4,\tau_0,\cdots,\tau_4$ be real parameters. The most general potential such that the ODSE can be transformed into 
the confluent Heun equation (\ref{CHE}) is given by
\begin{equation}\label{confheun_pot} 
V_{CH}(x)=\frac{y^2(y-1)^2}{R^2(y)}\left[\ddot{R}(y)+\frac{(2y-1)\dot{R}(y)+4R(y)}{y(y-1)}+
\frac{(2y-1)^2R(y)}{y^2(y-1)^2}-\frac{5}{4}\frac{(\dot{R}(y))^2}{R(y)}\right]-\frac{S(y)}{R(y)},
\end{equation}
where
\begin{eqnarray*}
R(y)&=&\sigma_0 y^2+\sigma_1 (y-1)^2+\sigma_2 y^2(y-1)^2+\sigma_3y^2(y-1)+\sigma_4y(y-1),\\
S(y)&=&\tau_0 y^2+\tau_1 (y-1)^2+\tau_2 y^2(y-1)^2+\tau_3y^2(y-1)+\tau_4y(y-1).
\end{eqnarray*}
Moreover, the coordinate transformation $y=y(x)$ satisfies the differential equation
\[
(y^{'})^2 =\frac{4y^2(y-1)^2}{R(y)}.
\]
\end{theorem}
\begin{proof}
The Bose invariant of the confluent Heun equation can be written as
\begin{equation}\label{BoseConfHeun}
I(y)=\frac{\lambda_0 y^2+\lambda_1 (y-1)^2+\lambda_2 y^2(y-1)^2+\lambda_3y^2(y-1)+\lambda_4y(y-1)+\lambda_5 y(y-1)^2}
{4y^2(y-1)^2}
\end{equation}
with
\[
\lambda_0=1-(1-\delta)^2,\quad
\lambda_1=1-(1-\gamma)^2,\quad
\lambda_2=-16p^2, 
\lambda_3=8p(2\alpha-\delta),\quad
\lambda_4=-2(2\sigma+\gamma\delta),\quad
\lambda_5=-8p\gamma. 
\]
On the other side (\ref{BoseConfHeun}) can also be expressed in terms of the $\lambda_i$'s as 
\[
I(y)=\frac{\lambda_2 y^4+(-2\lambda_2+\lambda_3+\lambda_5)y^3+(\lambda_0+\lambda_1+\lambda_2-\lambda_3+\lambda_4-2\lambda_5)
y^2+(-2\lambda_1-\lambda_4+\lambda_5)y+\lambda_1}
{4y^2(y-1)^2}.
\]
Moreover, $I$ is a rational function of the form
\[
I(y)=\frac{P(y)}{Q(y)}=\frac{c_0+c_1 y+c_2 y^2+c_3 y^3+c_4 y^4}{4y^2(y-1)^2}
\]
and it admits the decomposition $I(y)=I_1(y)k^2+I_0(y)$ if and only if $P(y)=R(y)k^2+S(y)$ for some polynomials 
$R(y)=a_0+a_1y+a_2y^2+a_3y^3+a_4y^4$ and $S(y)=b_0+b_1y+b_2y^2+b_3y^3+b_4y^4$. This will be the case if
\begin{eqnarray*}
c_0&=&a_0k^2+b_0=\lambda_1,\\
c_1&=&a_1k^2+b_1=-2\lambda_1-\lambda_4+\lambda_5,\\
c_2&=&a_2k^2+b_2=\lambda_0+\lambda_1+\lambda_2-\lambda_3+\lambda_4-2\lambda_5\\
c_3&=&a_3k^2+b_3=-2\lambda_2+\lambda_3+\lambda_5,\\
c_4&=&a_4k^2+b_4=\lambda_2. 
\end{eqnarray*}
From the above set of equations it is straightforward to see that $\lambda_0$, $\lambda_1$, $\lambda_2$, $\lambda_5-\lambda_4$, and 
$\lambda_3+\lambda_5$ must be linear in $k^2$. Let us rewrite the Bose invariant in terms of this quantities, that is
\[
 I(y)=\frac{\lambda_0 y^2+\lambda_1 (y-1)^2+\lambda_2 y^2(y-1)^2+(\lambda_3+\lambda_5)y^2(y-1)+(\lambda_4-
 \lambda_5)y(y-1)}{4y^2(y-1)^2}.
\]
By setting $\lambda_i=\sigma_ik^2+\tau_i$ with $i=0,1,2$ and $\lambda_5-\lambda_4=\sigma_3 k^2+\tau_3$, 
$\lambda_3+\lambda_5=\sigma_4 k^2+\tau_4$ we can rewrite (\ref{BoseConfHeun}) in the form $I(y)=I_1(y)k^2+I_0(y)$ with
\[
I_1(y)=\frac{R(y)}{4y^2(y-1)^2},\quad  I_0(y)=\frac{S(y)}{4y^2(y-1)^2},
\]
and
\begin{eqnarray*}
R(y)&=&\sigma_0 y^2+\sigma_1 (y-1)^2+\sigma_2 y^2(y-1)^2+\sigma_3y^2(y-1)+\sigma_4y(y-1),\\
S(y)&=&\tau_0 y^2+\tau_1 (y-1)^2+\tau_2 y^2(y-1)^2+\tau_3y^2(y-1)+\tau_4y(y-1).
\end{eqnarray*}
Taking into account that
\[
\dot{I}_1(y)=\frac{\dot{R}(y)}{4y^2(y-1)^2}-\frac{(2y-1)R(y)}{2y^3(y-1)^3},
\]
we find that
\begin{equation}\label{CI11}
\frac{(\dot{I}_1)^2}{I_1}=\frac{(\dot{R}(y))^2}{4y^2(y-1)^2 R(y)}+
\frac{(2y-1)^2R(y)}{y^4(y-1)^4}-\frac{(2y-1)\dot{R}(y)}{y^3(y-1)^3}.
\end{equation}
Moreover,
\begin{equation}\label{CI22}
\ddot{I}_1(y)=\frac{\ddot{R}(y)}{4y^2(y-1)^2}-\frac{(2y-1)\dot{R}(y)+R(y)}{y^3(y-1)^3} 
+\frac{3}{2}\frac{(2y-1)^2R(y)}{y^4(y-1)^4}.
\end{equation}
Finally, substituting (\ref{CI11}) and (\ref{CI22}) into (\ref{Milson_pot}) we obtain (\ref{confheun_pot}).~~$\square$
\end{proof}
We consider now the case of the biconfluent Heun equation \cite{Ronveaux}
\begin{equation}\label{BCHE}
y\frac{d^2 v}{dy^2}+\left(1+\alpha-\beta y-2y^2\right)\frac{dv}{dy}+\left[(\gamma-\alpha-2)y-\frac{\delta+(1+\alpha)\beta}{2}\right]v(y)=0,
\end{equation}
where $\alpha,\beta,\gamma,\delta$ are arbitrary parameters.
\begin{theorem}
Let $\beta_0,\cdots,\beta_3,\rho_0,\cdots,\rho_3$ be real parameters. The most general potential such that the ODSE can be 
transformed into the biconfluent Heun equation (\ref{BCHE}) is given by
\begin{equation}\label{bconfheun_pot} 
V_{BCH}(x)=\frac{y^2}{R^2(y)}\left[\ddot{R}(y)+\frac{\dot{R}(y)}{y}+
\frac{R(y)}{y^2}-\frac{5}{4}\frac{(\dot{R}(y))^2}{R(y)}\right]-\frac{S(y)}{R(y)},
\end{equation}
where
\[
R(y)=\beta_3 y^3+\beta_2 y^2+\beta_1 y+\beta_0,\quad S(y)=-4y^4+\rho_3 y^3+\rho_2 y^2+\rho_1 y+\rho_0.
\]
Moreover, the coordinate transformation $y=y(x)$ satisfies the differential equation
\[
(y^{'})^2 =\frac{4y^2}{R(y)}.
\]
\end{theorem}
\begin{proof}
The Bose invariant of the biconfluent Heun equation can be written as
\begin{equation}\label{BoseBConfHeun}
I(y)=\frac{-4y^4-4\beta y^3+(4\gamma-\beta^2)y^2-2\delta y+1-\alpha^2}{4y^2}
\end{equation}
and it can be decomposed as $I(y)=I_1(y)k^2+I_0(y)$ if
\[
\beta=-\frac{\beta_3}{4}k^2-\frac{\rho_3}{4},\quad 4\gamma-\beta^2=\beta_2k^2+\rho_2,\quad \delta=-\frac{\beta_1}{2}k^2-\frac{\rho_1}{2},
\quad 1-\alpha^2=\beta_0 k^2+\rho_0. 
\]
It is straightforward to verify that
\[
I_1(y)=\frac{R(y)}{4y^2},\quad  I_0(y)=\frac{S(y)}{4y^2}
\]
with
\[
R(y)=\beta_3 y^3+\beta_2 y^2+\beta_1 y+\beta_0,\quad S(y)=-4y^4+\rho_3 y^3+\rho_2 y^2+\rho_1 y+\rho_0.
\]
Considering that
\[
\dot{I}_1(y)=\frac{\dot{R}(y)}{4y^2}-\frac{R(y)}{2y^3},
\]
we find that
\begin{equation}\label{BCI11}
\frac{(\dot{I}_1)^2}{I_1}=\frac{(\dot{R}(y))^2}{4y^2 R(y)}+
\frac{R(y)}{y^4}-\frac{\dot{R}(y)}{y^3}.
\end{equation}
In addition, we have
\begin{equation}\label{BCI22}
\ddot{I}_1(y)=\frac{\ddot{R}(y)}{4y^2}-\frac{\dot{R}(y)}{y^3} 
+\frac{3}{2}\frac{R(y)}{y^4}.
\end{equation}
Finally, substituting (\ref{BCI11}) and (\ref{BCI22}) into (\ref{Milson_pot}) we obtain (\ref{bconfheun_pot}).~~$\square$
\end{proof}
We consider now the case of the double confluent Heun equation \cite{Ronveaux}
\begin{equation}\label{DCHE}
y^2\frac{d^2 v}{dy^2}+y\frac{dv}{dy}+\left(\frac{B_{-2}}{y^2}+\frac{B_{-1}}{y}+B_0+B_1y+B_2 y^2\right)v(y)=0,
\end{equation}
where $B_{-2},\cdots,B_2$ are arbitrary parameters.
\begin{theorem}
Let $\mu_0,\cdots,\mu_4,\nu_0,\cdots,\nu_4$ be real parameters. The most general potential such that the ODSE can be 
transformed into the double confluent Heun equation (\ref{DCHE}) is given by
\begin{equation}\label{dconfheun_pot} 
V_{DCH}(x)=\frac{y^4}{R^2(y)}\left[\ddot{R}(y)+\frac{2\dot{R}(y)}{y}-\frac{5}{4}\frac{(\dot{R}(y))^2}{R(y)}\right]-\frac{S(y)}{R(y)},
\end{equation}
where
\[
R(y)=\mu_4 y^4+\mu_3 y^3+\mu_2 y^2+\mu_1 y+\mu_0,\quad S(y)=\nu_4 y^4+\nu_3 y^3+\nu_2 y^2+\nu_1 y+\nu_0.
\]
Moreover, the coordinate transformation $y=y(x)$ satisfies the differential equation
\[
(y^{'})^2 =\frac{4y^4}{R(y)}.
\]
\end{theorem}
\begin{proof}
The Bose invariant of the double confluent Heun equation is
\begin{equation}\label{BoseDConfHeun}
I(y)=\frac{4B_2 y^4+4B_1 y^3+(4B_0+1)y^2+4B_{-1}y+4B_{-2}}{4y^4}
\end{equation}
and it can be decomposed as $I(y)=I_1(y)k^2+I_0(y)$ if
\[
B_2=\frac{\mu_4}{4}k^2+\frac{\nu_4}{4},\quad B_1=\frac{\mu_3}{4}k^2+\frac{\nu_3}{4},\quad 
4B_0+1=\mu_2 k^2+\nu_2, \quad B_{-1}=\frac{\mu_0}{4} k^2+\frac{\nu_1}{4},\quad
B_{-2}=\frac{\mu_0}{4}k^2+\frac{\nu_0}{4}. 
\]
It is straightforward to verify that
\[
I_1(y)=\frac{R(y)}{4y^4},\quad  I_0(y)=\frac{S(y)}{4y^4}
\]
with
\[
R(y)=\mu_4 y^4+\mu_3 y^3+\mu_2 y^2+\mu_1 y+\mu_0,\quad S(y)=\nu_4 y^4+\nu_3 y^3+\nu_2 y^2+\nu_1 y+\nu_0.
\]
If we take into account that
\[
\dot{I}_1(y)=\frac{\dot{R}(y)}{4y^4}-\frac{R(y)}{y^5},
\]
we arrive at
\begin{equation}\label{DCI11}
\frac{(\dot{I}_1)^2}{I_1}=\frac{(\dot{R}(y))^2}{4y^4 R(y)}+
\frac{4R(y)}{y^6}-\frac{2\dot{R}(y)}{y^5}.
\end{equation}
We also have
\begin{equation}\label{DCI22}
\ddot{I}_1(y)=\frac{\ddot{R}(y)}{4y^4}-\frac{2\dot{R}(y)}{y^5}+\frac{5R(y)}{y^6}.
\end{equation}
Finally, substituting (\ref{DCI11}) and (\ref{DCI22}) into (\ref{Milson_pot}) we obtain (\ref{dconfheun_pot}).~~$\square$
\end{proof}
We consider now the case of the triconfluent Heun equation \cite{Ronveaux}
\begin{equation}\label{TCHE}
\frac{d^2 v}{dy^2}+\left(A_0+A_1 y+A_2 y^2-\frac{9}{4}y^4\right)v(y)=0,
\end{equation}
where $A_{0},A_1,A_2$ are arbitrary parameters.
\begin{theorem}
Let $\vartheta_i,\omega_i$ with $i=0,1,2$ be real parameters. The most general potential such that the ODSE can be 
transformed into the triconfluent Heun equation (\ref{TCHE}) is given by
\begin{equation}\label{tconfheun_pot} 
V_{TCH}(x)=\frac{1}{R^2(y)}\left[2\vartheta_2-\frac{5}{4}\frac{(2\vartheta_2 y+\vartheta_1)^2}{R(y)}\right]-\frac{S(y)}{R(y)},
\end{equation}
where
\[
R(y)=\vartheta_2 y^2+\vartheta_1 y+\vartheta_0,\quad S(y)=-9 y^4+\omega_2 y^2+\omega_1 y+\omega_0.
\]
Moreover, the coordinate transformation $y=y(x)$ satisfies the differential equation
\[
(y^{'})^2 =\frac{4}{R(y)}.
\]
\end{theorem}
\begin{proof}
The Bose invariant of the triconfluent Heun equation is trivially given by
\begin{equation}\label{BoseTConfHeun}
I(y)=\frac{-9y^4+4A_2 y^2+4A_1 y +4A_0}{4}.
\end{equation}
The decomposition $I(y)=I_1(y)k^2+I_0(y)$ can be achieved if we choose
\[
A_2=\frac{\vartheta_2}{4}k^2+\frac{\omega_2}{4},\quad A_1=\frac{\vartheta_1}{4}k^2+\frac{\omega_1}{4},\quad 
A_0=\frac{\vartheta_0}{4}k^2+\frac{\omega_0}{4}. 
\]
Hence, we find that
\[
I_1(y)=\frac{R(y)}{4},\quad  I_0(y)=\frac{S(y)}{4}
\]
with
\[
R(y)=\vartheta_2 y^2+\vartheta_1 y+\vartheta_0,\quad S(y)=-9 y^4+\omega_2 y^2+\omega_1 y+\omega_0.
\]
Substituting
\[
\frac{(\dot{I}_1)^2}{I_1}=\frac{(\dot{R}(y))^2}{4R(y)},\quad
\ddot{I}_1(y)=\frac{\ddot{R}(y)}{4} 
\]
into (\ref{Milson_pot}) we obtain (\ref{tconfheun_pot}).~~$\square$
\end{proof}
We conclude this section by deriving the most general potential such that the ODSE can be transformed into the generalized Heun 
equation (\ref{GHE}). Since this equation contains the Heun equation as a special case, the potential derived in the next result which 
depends upon $14$ real parameters will also generalize the class of potentials obtained in Theorem~\ref{potenzialeHeun}.
\begin{theorem}\label{potenzialeGHeun}
Let $\eta_1,\cdots,\eta_7,\xi_1,\cdots,\xi_7$ be real parameters. The most general potential such that the ODSE can be transformed 
into the generalized Heun equation (\ref{GHE}) is given by
\begin{equation}\label{gheun_pot} 
V_{GH}(x)=\frac{y^2(y-1)^2(y-a)^2}{R^2(y)}\left[\ddot{R}(y)+\frac{G(y)\dot{R}(y)-2R(y)\dot{G}(y)}{y(y-1)(y-a)}+
\frac{R(y)G^2(y)}{y^2(y-1)^2(y-a)^2}-\frac{5}{4}\frac{(\dot{R}(y))^2}{R(y)}\right]-\frac{S(y)}{R(y)},,
\end{equation}
where $G(y)=3y^2-2(a+1)y+a$ and
\[
R(y)=\eta_1(y-1)^2(y-a)^2+\eta_2 y^2(y-a)^2+\eta_3y^2(y-1)^2+\eta_4y^2(y-1)^2(y-a)^2
\]
\[
+\eta_5 y(y-1)^2(y-a)^2+\eta_6y^2(y-1)(y-a)+\eta_7 y(y-1)(y-a),
\]
\[
S(y)=\xi_1(y-1)^2(y-a)^2+\xi_2 y^2(y-a)^2+\xi_3y^2(y-1)^2+\xi_4y^2(y-1)^2(y-a)^2
\]
\[
+\xi_5 y(y-1)^2(y-a)^2+\xi_6y^2(y-1)(y-a)+\xi_7 y(y-1)(y-a),
\]
Moreover, the coordinate transformation $y=y(x)$ satisfies the differential equation
\[
(y^{'})^2 =\frac{4y^2(y-1)^2(y-a)^2}{R(y)}.
\]
\end{theorem}
\begin{proof}
Let us write the Bose invariant of the generalized Heun equation into the following form
\[
4y^2(y-1)^2(y-a)^2I(y)=\lambda_0 y(y-1)(y-a)+\lambda_1 y^2(y-1)(y-a)+\lambda_2 y(y-1)(y-a)^2+\lambda_3y(y-1)^2(y-a)+
\lambda_4y(y-1)^2(y-a)^2
\]
\[
+\lambda_5 y^2(y-1)(y-a)^2+\lambda_6 y^2(y-1)^2(y-a)+\lambda_7(y-1)^2(y-a)^2+\lambda_8 y^2(y-a)^2+\lambda_9 y^2(y-1)^2 
\]
\begin{equation}\label{BoseGHeun}
\lambda_{10}y^2(y-1)^2(y-a)^2+\lambda_{11}y^3(y-1)(y-a) 
\end{equation}
with
\[
\lambda_0=4\beta_0,\quad
\lambda_1=4\beta_1-2\alpha_1\alpha_2,\quad
\lambda_2=-2\alpha_0\alpha_1,\quad
\lambda_3=-2\alpha_0\alpha_2,\quad
\lambda_4=-2\alpha_0\alpha_3,\quad
\lambda_5=-2\alpha_1\alpha_3,
\]
\[
\lambda_6=-2\alpha_2\alpha_3,\quad
\lambda_7=2\alpha_0-\alpha_0^2,\quad
\lambda_8=2\alpha_1-\alpha_1^2,\quad
\lambda_9=2\alpha_2-\alpha_2^2,\quad
\lambda_{10}=-\alpha_3^2,\quad
\lambda_{11}=4\beta_2. 
\]
The Bose invariant is a rational function of the form
\[
I(y)=\frac{P(y)}{Q(y)}=\frac{c_0+c_1 y+c_2 y^2+c_3 y^3+c_4 y^4+c_5 y^5+c_6 y^6}{4y^2(y-1)^2(y-a)^2}
\]
and it admits the decomposition $I(y)=I_1(y)k^2+I_0(y)$ if and only if $P(y)=R(y)k^2+S(y)$ for some polynomials 
$R(y)=\eta_0+\eta_1y+\eta_2y^2+\eta_3y^3+\eta_4y^4+\eta_5 y^5+\eta_6 y^6$ and 
$S(y)=\xi_0+\xi_1y+\xi_2y^2+\xi_3y^3+\xi_4y^4+\xi_5 y^5+\xi_6 y^6$. This will be the case if
\[
c_0=a_0k^2+b_0=a^2\lambda_7,\quad
c_1=a_1k^2+b_1=a\lambda_0-a^2\lambda_2-a\lambda_3+a^2\lambda_4-(2a+a^2)\lambda_7,
\]
\[
c_2=a_2k^2+b_2=-(1+a)\lambda_0+a\lambda_1+(2a+a^2)\lambda_2+(1+2a)\lambda_3-(2a+2a^2)\lambda_4-a^2\lambda_5-a\lambda_6
+(1+4a+a^2)\lambda_7+a^2\lambda_8+\lambda_9+a^2\lambda_{10},
\]
\[
c_3=a_3k^2+b_3=\lambda_0-(1+a)\lambda_1-(1+2a)\lambda_2-(2+a)\lambda_3+(1+4a+a^2)\lambda_4
\]
\[
+(2a+a^2)\lambda_5+(1+2a)\lambda_6
-(2+2a)\lambda_7-2a\lambda_8-2\lambda_9-(2a+2a^2)\lambda_{10}+a\lambda_{11},
\]
\[
c_4=a_4k^2+b_4=\lambda_1+\lambda_2+\lambda_3-(2+2a)\lambda_4-(1+2a)\lambda_5-(2+a)\lambda_6+\lambda_7+\lambda_8
+\lambda_9+(1+4a+a^2)\lambda_{10}-(1+a)\lambda_{11},
\]
\[
c_5=a_5k^2+b_5=\lambda_4+\lambda_5+\lambda_6-(2+2a)\lambda_{10}+\lambda_{11},\quad
c_6=a_6k^2+b_6=\lambda_{10}.
\]
It can be immediately seen that $\lambda_7$ and $\lambda_{10}$ are linear in $k^2$. By applying the Gauss-Jordan elimination method to 
the remaining equations involving the $\lambda_i$'s we discover that $\lambda_8$, $\lambda_9$, 
$\lambda_4+\lambda_5+\lambda_6+\lambda_{11}$, $\lambda_1+\lambda_2+\lambda_3+\lambda_5+a\lambda_6+(1+a)\lambda_{11}$, and $\lambda_0
-\lambda_3-a(\lambda_2+\lambda_5+\lambda_6+\lambda_{11})$ must be also linear in $k^2$. At this point the Bose invariant 
(\ref{BoseGHeun}) can equivalently be written as
\[
4y^2(y-1)^2(y-a)^2I(y)=\lambda_7(y-1)^2(y-a)^2+\lambda_8 y^2(y-a)^2+\lambda_9y^2(y-1)^2+\lambda_{10}y^2(y-1)^2(y-a)^2
\]
\[
+(\lambda_4+\lambda_5+\lambda_6+\lambda_{11})y(y-1)^2(y-a)^2+[\lambda_1+\lambda_2+\lambda_3+\lambda_5+a\lambda_6+(1+a)\lambda_{11}]y^2(y-1)(y-a)
\]
\[
[\lambda_0
-\lambda_3-a(\lambda_2+\lambda_5+\lambda_6+\lambda_{11})]y(y-1)(y-a). 
\]
Finally, if we set $\lambda_7=\eta_1 k^2+\xi_1$, $\lambda_8=\eta_2 k^2+\xi_2$, $\lambda_9=\eta_3 k^2+\xi_3$, 
$\lambda_{10}=\eta_4 k^2+\xi_4$, $\lambda_4+\lambda_5+\lambda_6+\lambda_{11}=\eta_5 k^2+\xi_5$, 
$\lambda_1+\lambda_2+\lambda_3+\lambda_5+a\lambda_6+(1+a)\lambda_{11}=\eta_6 k^2+\xi_6$, and $\lambda_0
-\lambda_3-a(\lambda_2+\lambda_5+\lambda_6+\lambda_{11})=\eta_7k^2+\xi_7$, where the $\eta_i,\xi_i$'s can be expressed in terms of the 
$a_i,b_i$'s introduced before, we can rewrite the Bose invariant as $I(y)=I_1(y)k^2+I_0(y)$ where
\[
I_1(y)=\frac{R(y)}{4y^2(y-1)^2(y-a)^2},\quad  I_0(y)=\frac{S(y)}{4y^2(y-1)^2(y-a)^2},
\]
and
\[
R(y)=\eta_1(y-1)^2(y-a)^2+\eta_2 y^2(y-a)^2+\eta_3y^2(y-1)^2+\eta_4y^2(y-1)^2(y-a)^2
\]
\[
+\eta_5 y(y-1)^2(y-a)^2+\eta_6y^2(y-1)(y-a)+\eta_7 y(y-1)(y-a),
\]
\[
S(y)=\xi_1(y-1)^2(y-a)^2+\xi_2 y^2(y-a)^2+\xi_3y^2(y-1)^2+\xi_4y^2(y-1)^2(y-a)^2
\]
\[
+\xi_5 y(y-1)^2(y-a)^2+\xi_6y^2(y-1)(y-a)+\xi_7 y(y-1)(y-a),
\]
At this point the derivation of the potential is the same as in Theorem~\ref{potenzialeHeun}~~$\square$
\end{proof}

\section{Conclusions}
The Schr\"odinger equation is one of the key equations in
Quantum Mechanics and plays a dominant role in many branches
of science and technology. The history of analytical
solvability of this equations has, of course, taken different
routes. The concept of shape invariance is one example of
a successful strategy \cite{shape}.
Another line of attacking the problem is to transform
a well known second order ordinary differential equation into
the ODSE obtaining hereby a class of potentials.
Such a procedure has been followed several times in the
literature, mostly with the hypergeometric equation as a
starting point. The Heun equation, its confluent forms and the
generalized Heun equation offer the next
possible generalization in this procedure. The class of Heun
potentials allowing a transformation from the ODSE to one
of the Heun equations has been derived in the paper. 
To appreciate how general these potentials are
we mention that the potentials of Iwata class
have five free parameters. The Natanzon class
contains six free parameters. In comparison, the
Heun class has as many as ten free parameters
(with the exception of the biconfluent Heun class which has eight
parameters). The class of potentials of the
generalized Heun equation has even fourteen free parameters.
This fact will increase the number of
possible applications.
It is therefore worthwhile to extend the present
study to consider the eigenvalue problem for bound states and other 
related
problems. We intend to cover this and other topics 
in a forthcoming publication. 

\begin{acknowledgments}
One of the authors (D.B.) thanks Prof. Gian Michele Graf for the fruitful and stimulating discussions he had during his visit at the 
Institute of Theoretical Physics, ETH Zurich, Switzerland.
\end{acknowledgments}


\end{document}